 \documentclass[journal]{IEEEtran}

\usepackage{xspace}
\usepackage{mathrsfs}
\usepackage{amsopn}

\newcommand{\Bc}{\mathcal{B}}

\newcommand{\Ec}{\mathcal{E}}

\newcommand{\Sc}{\mathcal{S}}

\newcommand{\Vc}{\mathcal{V}}

\newcommand{\Yc}{\mathcal{Y}}
\newcommand{\Zc}{\mathcal{Z}}

\renewcommand{\Pr}{\mathscr{P}}

\newcommand{\aep}{{\mathcal{T}_{\epsilon}^{(n)}}}

\newcommand{\Sh}{{\hat{S}}}

\newcommand{\Vh}{{\hat{V}}}

\newcommand{\Yh}{{\hat{Y}}}

\newcommand{\lh}{{\hat{l}}}

\newcommand{\uh}{{\hat{u}}}
\newcommand{\vh}{{\hat{v}}}

\newcommand{\yh}{{\hat{y}}}

\newcommand{\Vt}{{\tilde{V}}}

\newcommand{\vt}{{\tilde{v}}}

\def\d{\delta}
\def\e{\epsilon}

\DeclareMathOperator\E{E}
\let\P\relax
\DeclareMathOperator\P{P}

\newcommand{\Bern}{\mathrm{Bern}}

\newcommand{\U}{\mathcal{U}}

\def\textiid{i.i.d.\@\xspace}
\newcommand\iid{\ifmmode\text{ i.i.d. } \else \textiid \fi}

\usepackage{latexsym}
\usepackage{cite}
\usepackage{amssymb}
\usepackage{bm}
\usepackage{amsmath, amssymb}
\usepackage{graphicx}
\usepackage{epsfig}
\usepackage{psfrag}
\usepackage[dvips]{color}
\usepackage{microtype}
\DisableLigatures{encoding = *, family = * }
\usepackage{version}

\usepackage{theorem}
\theoremheaderfont{\bfseries\upshape} \theoremstyle{plain}

\theorembodyfont{\rmfamily}

\theoremheaderfont{\itshape}

\theorembodyfont{\itshape}
\newtheorem{theorem}{Theorem}

\newtheorem{lemma}{Lemma} 
\newtheorem{corollary}{Corollary}

\def\Pr{ {\mathrm{Pr}}}
\def\cV{ {\mathcal{V}}}

\begin{document}

\title{Compression with Actions}
\author{Lei Zhao, Yeow-Khiang Chia and Tsachy Weissman
\thanks{ Lei Zhao was with Stanford University when this work was done. He is now with Jump Operations LLC. Email: zhaolei122@alumni.stanford.edu}
\thanks{Yeow-Khiang Chia was with Stanford University when this work was done. He is now with Institute for Infocomm Research, Singapore. Email: yeowkhiang@gmail.com}  
\thanks{Tsachy Weissman is with Stanford University. Email: tsachy@stanford.edu}
\thanks{Material presented in part at Allerton Conference on Communications, Control and Computing, 2011.}
}

\maketitle


\begin{abstract}
We consider the setting where actions can be used to modify a state sequence before compression. The minimum rate needed to losslessly describe the optimal modified sequence  is characterized when the state sequence is either non-causally or causally available at the action encoder.  The achievability is closely related to the optimal channel coding strategy for channel with states. We also extend the analysis to the the lossy case.
\end{abstract}
\section{Introduction}
Consider the standard Shannon-theoretic lossy source coding setting where we have a source $S^n$ that we wish to perform lossy compression on. The encoder receives the source $S^n$ and produces an index $M$ that is sent to the decoder. Based on the index, the decoder produces a lossy reconstruction, $\Sh^n$, such that the per symbol distortion constraint is satisfied. An alternative view of this problem is as one where the encoder is first required to produce the reconstruction sequence $\Sh^n$ and then uses a lossless compression algorithm to describe $\Sh^n$ to the decoder. This point of view on lossy source coding, depicted in Figure \ref{fig:lossysc}, has been instrumental in recent developments of an approach to universal and implementable lossy compressors, cf.~\cite{SMW2012},~\cite{SW2012} and references therein.

In this paper, we generalize the above setting by asking the following question: what if the encoder makes an ``error'' in outputting the reconstruction sequence? The encoder may wish to take ``actions'' to output a sequence of reconstruction symbols. However, due to noise, the reconstruction symbols at the output of the encoder may be different from the intended reconstruction symbols. In this case, we are still interested in sending the reconstruction sequence (the \textit{modified} source sequence) to the decoder. The question then is, what is the optimal rate-distortion tradeoff in such a scenario? As a more concrete example, consider lossy compression of a binary source $S^n$. With the source as input, the encoder first attempts to output the desired reconstruction sequence, but due to errors in the circuitry of the encoder, a bit that is meant to be one can still be zero with some probability and vice versa. Using a universal lossless compression algorithm, we transmit the output of the ``faulty'' encoder, $\Sh^n$, to the decoder. We are now interested in the optimum rate-distortion tradeoff under the assumption of a ``faulty'' encoder.
\begin{figure}[h]{
\psfrag{pyas}[c]{\small $\P_{Y|A,S}$}
\psfrag{s}[.2]{\small $S^n$}
\psfrag{sh}[c]{\small $\Sh^n$}
\psfrag{m}[c]{\small $2^{nR}$}
\psfrag{se}[c]{\small Encoder}
\psfrag{dec}[c]{\small Decoder}
\psfrag{ulc}[c]{\small Univ. lossless}
\centerline{\includegraphics[width=3.35in]{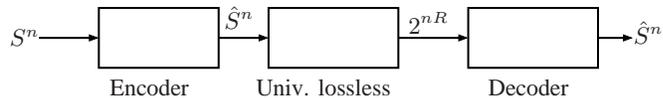}}} 
\caption{Lossy source coding.}\label{fig:lossysc}
\label{fig.lossysc}
\end{figure}

As another example of our general setting, which may seem at first sight to be unrelated to the question we asked above, imagine that we have a number of robots working on a factory floor and the positions of all the robots need to be reported to a remote location. Letting $S$ represent the position of a robot, we would expect to send $H(S)$ bits to the remote location. However, what if the robots can take {\em actions} to change their positions so that they can be more efficiently described? A local command center can first give commands (actions) to the robots so that they move in a cooperative way into a final position sequence that requires fewer bits to describe. The command center may face two issues in general: cost constraints and uncertainty. A cost constraint occurs because each robot should save its power and not move too far away from its current location. The uncertainty is a result of the robots not moving exactly as instructed by the local command center. 


Both examples are instantiations of the problem setting illustrated in Fig.~\ref{fig.setup1} (Formal definitions are given in the next section). Here, $S^n$ is our observed source (or state) sequence. We assume a general cost function $\Lambda(a,s,y)$ and  a general relation, specified by a conditional PMF $p(y|a,s)$, relating the modified source sequence to be compressed to the original source sequence (state) and action taken by the encoder toward modifying it. As shown in the preceding examples, we are interested in compressing the final output $Y^n$. 
\begin{figure}[h]{
\psfrag{pyas}[c]{\small $\P_{Y|A,S}$}
\psfrag{s}[.2]{\small $S^n$}
\psfrag{y}[c]{\small $Y^n$}
\psfrag{r}[c]{$\small 2^{nR}$}
\psfrag{z}[c]{\small $Z^n$}
\psfrag{yr}[c]{\small $\Yh^n$}
\psfrag{a}[c]{\small $A^n$}
\psfrag{en}[c]{\small Compressor}
\psfrag{dec}[c]{\small Decoder}
\psfrag{ac}[l]{\small Action Encoder}
\centerline{\includegraphics[width=3.7in]{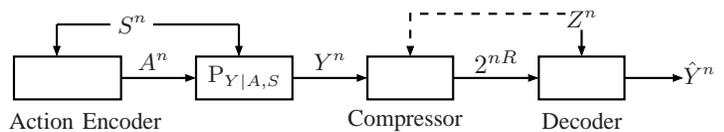}}}
\caption{Compression with actions. The Action encoder first observes the state sequence $S^n$ and then generates an action sequence $A^n$. The $i$th output $Y_i$ is the output of a channel $p(y|a,s)$ when $a=A_i$ and $s = S_i$. The compressor outputs a description of $Y^n$, $M\in [1:2^{nR}]$, from $Y^n$ alone if the side information $Z^n$ is not available at the compressor. If the side information is available, then the compressor generates the description based on $Y^n$ and $Z^n$. The remote decoder generates $\hat{Y}^n$ based on $M$ and its available side information $Z^n$ as a reconstruction of $Y^n$.}
\label{fig.setup1}
\end{figure}

Our problem setup is also closely related to the channel coding problem when the state information is available at the encoder. The case where the state information is causally available was first solved by Shannon in \cite{Shannon1958}. When the state information is non-causally known at the encoder, the channel capacity result was derived in \cite{GP1980} and \cite{HG1983}. Various interesting extensions can be found in \cite{SC2002,SK2005,KSC2008,SS2009,GWS2010}. The difference in our approach described here is that we make the output of the channel as compressible as possible. 
Our main results when the decoder requires lossless reconstruction are given in section \ref{sect:3}, where we characterize the rate-cost tradeoff function for the setting in Fig. \ref{fig.setup1}. We also characterize the rate-cost function when $S^n$ is only causally known at the action encoder. In section \ref{sect:4}, we extend the setting to the lossy case where the decoder requires a lossy version of $Y^n$.{We characterize the rate-distortion cost function when $S^n$ is causally known at the action encoder and the side information $Z^n$ is available at both the compressor and the decoder. For other settings, we give achievable schemes for the rate-distortion cost functions. We conclude in Section \ref{sect:5}, where we mention some possible extensions for future consideration.}

\section{Definitions} \label{sect:2}
We give formal definitions for the setups under consideration in this section. We will follow the notation of \cite{El-Gamal--Kim2010}. Sources $(S^n, Z^n)$ are assumed to be i.i.d.; i.e. $(S^n, Z^n) \sim \prod_{i=1}^n p_{S,Z}(s_i,z_i)$.
\subsection{Lossless case with no side information at the compressor}
We now give the definitions for the case when the side information $Z^n$ is not available at the compressor. Referring to Figure \ref{fig.setup1}, a $(n, 2^{nR})$ code for this setup consists of
\begin{itemize}
  \item an action encoding function $f_a: \mathcal{S}^n \to \mathcal{A}^n$;
  \item a compression function $f_c: \mathcal{Y}^n \to M \in [1:2^{nR}]$;
  \item a decoding function $f_d: [1:2^{nR}]\times \mathcal{Z}^n \to \hat{Y}^n$.
\end{itemize}
The average cost of the system is $\E\Lambda(A^n,S^n,Y^n)\triangleq \frac{1}{n}\sum_{i=1}^n \E\Lambda(A_i,S_i,Y_i)$. A rate-cost tuple $(R,B)$ is said to be achievable if there exists a sequence of codes such that
\begin{align}
\limsup_{n\rightarrow\infty}\mathrm{Pr}( Y^n\neq f_d( f_c(Y^n), Z^n) )=0, \label{eqn:pe}\\
\limsup_{n\rightarrow\infty} \E\Lambda(A^n,S^n,Y^n)\leq B, \label{eqn:cost}
\end{align}
where $\Lambda(A^n,S^n,Y^n) = \sum_{i=1}^n \Lambda(A_i,S_i,Y_i)/n$. Given cost $B$, the {\em rate-cost function}, $R(B)$, is then the infimum of rates $R$ such that $(R,B)$ is achievable.

BC

\textit{Remark:} Suppose that the channel is given by $\P_{Y|A,S} = 1_{Y = A}$ (where $1_{(.)}$ is the indicator function) and that the cost constraint is given by $\Lambda(A^n, S^n)$, then we recover the standard lossy source coding setting with $A^n$ being the reconstruction sequence. 

EC

\subsection{Lossless case when side information is available at the compressor}
In the case when side information $Z^n$ is available at the compressor, the definitions remain mostly the same, with the exception that the compression function is now given by
\begin{align*}
f_c: \Yc^n \times \Zc^n \to M \in [1:2^{nR}].
\end{align*}
\subsection{Lossy case}
In the setting where the decoder requires a lossy version of $Y^n$, the definitions remain largely the same, with the exception that the probability of error constraint, inequality (\ref{eqn:pe}), is replaced by the following distortion constraint.
\begin{align}
\limsup_{n \to \infty}\E d(Y^n,\hat{Y}^n) = \limsup_{n\to \infty}\frac{1}{n}\sum_i^n \E d(Y_i,\hat{Y}_i)\leq D. \label{eqn:dist}
\end{align}
A rate $R$ is said to be achievable if there exists a sequence of $(n, 2^{nR})$ codes satisfying both the cost constraint (inequality \ref{eqn:cost}) and the distortion constraint (inequality \ref{eqn:dist}). Given cost $B$ and distortion $D$, the {\em rate-cost-distortion function}, $R(B,D)$, is then the infimum of rates $R$ such that the tuple $(R,B,D)$ is achievable.
\subsection{Causal observations of state sequence}
In both the lossless and lossy case, we will also consider the setup when the state sequence is only causally known at the action encoder. The definitions remain the same, except for the action encoding function which is now restricted to the following form: For each $i\in [1:n]$, $f_{a,i}: \mathcal{S}^i \to \mathcal{A}$.

\section{Lossless case} \label{sect:3}
In this section, we present our main results for the lossless case. For the lossless case, we will only consider the case when the side information is no available at the compressor, as it will be clear from the results that the presence of side information at the compressor does not change the rate-cost regions for both the case when $S^n$ is causally known, and the case when $S^n$ is non-causally known.

Theorem \ref{thm:1} gives the rate-cost function when the state sequence is non-causally available at the action encoder, while Theorem \ref{thm:2} gives the rate-cost function when the state sequence is causally available.
\subsection{Lossless, non-causal compression with action}
\begin{theorem}[Rate-cost function for lossless, non-causal case] \label{thm:1} The rate-cost function for the compression with action setup when state sequence $S^n$ is non-causally available at the action encoder is given by
\begin{equation}\label{eq.RC}
R(B)=\min_{ p(v|s), a=f(s,v): E\Lambda(S,A,Y)\leq B} I(V;S|Z)+H(Y|V,Z),
\end{equation}
where the joint distribution is of the form $p(z,s,v,a,y) = p(z,s)p(v|s)1_{\{f(s,v)=a\}}p(y|a,s)$. The cardinality of the auxiliary random variable $V$ is upper bounded by $|\Vc| \le |\Sc| +2$.
\end{theorem}
\textit{Remarks}
\begin{itemize}
\item Replacing $a=f(s,v)$ by a general distribution $p(a|s,v)$ does not decrease the minimum in (\ref{eq.RC}). For any joint distribution $p(s)p(s|v)p(a|s,v)$, we can always find a random variable $W$ and a function $f$ such that $W$ is independent of $S,V$ and $Y$, and $A = f(V,W,X)$. Consider $V' = (V,W)$. The Markov condition $ V'-(A,S)-(Y,Z)$ still holds. Thus $H(Y|V',Z)+I(V';S|Z)$ is achievable. Furthermore,
\begin{align*}
&I(V';S|Z)+H(Y|V', Z) \\&= I(V,W ; S|Z) + H(Y|V,W,Z) \\
                &\leq I(V,W; S|Z) + H(Y|V,Z)  \\
                &= I(V;S|Z)+H(Y|V,Z).
\end{align*}
  \item $R(B)$ is a convex function in $B$. 
  \item For each cost function $\Lambda(s,a,y)$, we can replace it with a new cost function involving only $s$ and $a$ by defining $\Lambda'(s,a)= E[\Lambda(S,A,Y)|S=s,A=a]$. Note that $Y$ is distributed as $p(y|s,a)$ given $S=s,A=a$.
  
  BC
  
  \item If we set $\P_{Y|A,S} = 1_{Y=A}$ and $Z = \emptyset$, then we have
  \begin{align*}
  I(V;S) + H(A|V) &= I(V,A;S) - I(A;S|V) + H(A|V) \\
  & = I(V,A;S) \\
  & \ge I(A;S).
  \end{align*}
  The rate-cost function then works out to
  \begin{align*}
  R(B) \ge \min_{\Lambda(A,S) \le B} I(A;S)
  \end{align*}
  for some $p(a|s)$. This recovers the standard lossy source coding result with $A$ being the reconstruction alphabet and $B$ being the desired distortion.
  
  EC
\end{itemize}
Achievability of Theorem \ref{thm:1} involves an interesting observation in the decoding operation, but before proving the theorem, we first state a corollary of Theorem \ref{thm:1}, the case when side information is absent ($Z = \emptyset$). We will also sketch an alternative achievability proof for the corollary, which will serve as a contrast to the achievability scheme for Theorem \ref{thm:1}.

\begin{corollary}[Side information is absent]
If $Z = \emptyset$, then rate-cost function is given by
\begin{align*}
R(B)=\min_{ p(v|s), a=f(s,v): E\Lambda(S,A,Y)\leq B} I(V;S)+H(Y|V)
\end{align*}
for some $p(s,v,a,y) = p(s)p(v|s)1_{\{f(s,v)=a\}}p(y|a,s)$.
\end{corollary}

\subsection*{Achievability sketch for Corollary 1}
Code book generation: Fix $p(v|s)$ and $f(s,v)$ and $\epsilon>0$.
\begin{itemize}
  \item  Generate $2^{n (I(S;V) + \epsilon)}$ $v^n(l)$ sequences independently, $l \in [1:2^{n(I(V;S) + \e)}]$, each according to $\prod p_V(v_i)$ to cover $S^n$.
  \item For each $V^n$ sequence, the $Y^n$ sequences that are jointly typical with $V^n$ are indexed by $2^{(n( H(Y|V) + \epsilon)}$ numbers.
\end{itemize}
Encoding and Decoding:
\begin{itemize}
  \item  The action encoder looks for a $V^n$ in the code book that is jointly typical with $S^n$ and generates $A_i = f(S_i,V_i), i=1,...,n$.
  \item The compressor looks for a $\Vh^n$ in the codebook that is jointly typical with the channel output $Y^n$ and sends the index of that $\Vh^n$ sequence to the decoder. The compressor then sends the index of $Y^n$ as described in the second part of code book generation.
    \item The decoder simply uses both indices from the compressor to reconstruct $Y^n$.
\end{itemize}

Using standard typicality arguments, we can show that the encoding succeeds with high probability and the probability of error can be made arbitrarily small.

{\em Remark:} Note that the index of $\Vh^n$ is not necessarily equal to $V^n$. That is, the $V^n$ codeword chosen by the action encoder can be different from the $\Vh^n$ codeword chosen by the compressor. But this is not an error event since we still recover the same $Y^n$ even if a different $V^n$ codeword was used.

This scheme, however, does not extend to the case when side information is available at the decoder. The term $H(S|Z,V)$ in Theorem \ref{thm:1} requires us to bin the set of $Y^n$ sequences according to the side information available at the decoder. If we were to extend the above achievability scheme, we would bin the set of $Y^n$ sequences to $2^{n(H(Y|Z,V)+\e)}$ bins. The compressor would find a $\Vh^n$ sequence that is jointly typical with $Y^n$, send the index to the decoder using a rate of $I(V;S|Z) + \e$, and then send the index of the bin which contains $Y^n$. The decoder would then look for the unique $Y^n$ sequence in the bin that is jointly typical with $\Vh^n$ and $Z^n$. Unfortunately, while the $\Vh^n$ codeword is jointly typical with $Y^n$ with high probability, it is not necessarily jointly typical with $Z^n$, since $\Vh^n$ may not be equal to $V^n$ ($V^n$ is jointly typical with $Z^n$ with high probability as $V^n$ is jointly typical with $S^n$ with high probability and $V-S-Z$). One could try to overcome this problem by insisting that the compressor finds the same $V^n$ sequence as the action encoder, but this requirement imposes additional constraints on the achievable rate.

Instead of requiring that the compressor finds a jointly typical $V^n$ sequence, we use an alternative approach to prove Theorem \ref{thm:1}. We simply bin the set of all $Y^n$ sequences to $2^{n(I(V;S|Z)+H(Y|Z,V)+\e)}$ bins and send the bin index to the decoder. The decoder looks for the \textit{unique} $Y^n$ sequence in bin $M$ such that $(V^n(l), Y^n, Z^n)$ are jointly typical for \textit{some} $l \in [1:2^{n(I(V;S) + \e)}]$. Note that there can more than one $V^n(l)$ sequence which is jointly typical with $(Y^n, Z^n)$, but this is not an error event as long as the $Y^n$ sequence in bin $M$ is unique. We now give the details of this achievability scheme.

\subsection*{Proof of achievability for Theorem \ref{thm:1}}
\subsection*{Codebook generation}
\begin{itemize}
\item Generate $2^{n(I(V;S) + \d(\e))}$ $V^n$ codewords according to $\prod_{i=1}^n p(v_i)$
\item For the entire set of possible $Y^n$ sequences, bin them uniformly at random to $2^{nR}$ bins, where $R > I(V;S) - I(V;Z) + H(Y|Z,V)$, $\Bc(M)$.
\end{itemize}
\subsection*{Encoding}
\begin{itemize}
\item Given $s^n$, the encoder looks for a $v^n$ sequence in the codebook such that $(v^n, s^n) \in \aep$. If there is more than one, it randomly picks one from the set of typical sequences. If there is none, it picks a random index from $[1:2^{nI(V;S)+ \d(\e)}]$.
\item It then generates $a^n$ according to $a_i = f(v_i,s_i)$ for $i\in [1:n]$.
\item At the second encoder, it takes the output $y^n$ sequences and sends out the bin index $M$ such that $y^n \in \Bc(M)$.
\end{itemize}
\subsection*{Decoding}
\begin{itemize}
\item The decoder looks for the \textit{unique} $\yh^n$ sequence such that $(v^n(l), \yh^n, z^n) \in \aep$ for \textit{some} $l \in [1:2^{n(I(V;S))}]$ and $\yh^n \in \Bc(M)$. If there is none or more than one, it declares an error.
\end{itemize}
\subsection*{Analysis of probability of error}
Define the following error events
\begin{align*}
\Ec_0 &:= \{(V^n(L), Z^n, Y^n) \notin \aep\}, \\
\Ec_l &:= \{(V^n(l), Z^n, \Yh^n) \in \aep \\
& \qquad \mbox{ for some } \Yh^n \neq Y^n, \Yh^n \in \Bc(M)\}.
\end{align*}

By symmetry of the codebook generation, it suffices to consider $M = 1$. The probability of error is upper bounded by
\begin{align*}
\P(\Ec) \le \P(\Ec_0) + \sum_{l=1}^{2^{n(I(V;S) + \d(\e))}} \P(\Ec_l).
\end{align*}
$\P(\Ec_0) \to 0$ as $n \to \infty$ following standard analysis of probability of error. It remains to analyze the second error term. Consider $\P(\Ec_l)$ and define $\Ec_l(V^n, Z^n): = \{(V^n(l), Z^n, \Yh^n) \in \aep \mbox{ for some } \Yh^n \neq Y^n, \Yh^n \in \Bc(1)\}$. We have{\allowdisplaybreaks
\begin{align*}
\P(\Ec_l) &= \P(\Ec_l(V^n,Z^n)) \\
& = \sum_{(v^n, z^n) \in \aep} \P(V^n(l) = v^n, Z^n = z^n)\P(\Ec_l(v^n,z^n)|v^n, z^n) \\
& = \sum_{(v^n, z^n) \in \aep} \left( \P(V^n(l) = v^n, Z^n = z^n)\mathbf{.} \right.\\
&\left. \qquad  \sum_{y^n}\P(Y^n = y^n|v^n, z^n)\P(\Ec_l(v^n,z^n)|v^n, z^n, y^n)\right) \\
& \stackrel{(a)}{\le}  \sum_{(v^n, z^n) \in \aep} \left( \P(V^n(l) = v^n, Z^n = z^n)\mathbf{.} \right.\\
&\left. \qquad \sum_{y^n}\P(Y^n = y^n|v^n, z^n) 2^{n(H(Y|Z,V) + \d(\e) -R)}\right) \\
& \stackrel{(b)}{=} \sum_{(v^n,z^n) \in \aep} \left(\P(V^n(l) = v^n)\P(Z^n = z^n) \mathbf{.} \right. \\
&\left. \qquad \qquad \qquad \qquad  2^{n(H(Y|Z,V) + \d(\e) -R)}\right) \\
& \le \left(2^{n(H(V,Z) + \d(\e))}2^{-n(H(V) - \d(\e))}2^{-n(H(Z) - \d(\e))}\mathbf{.}\right.\\
 &\left. \qquad \qquad 2^{n(H(Y|Z,V) + \d(\e) -R)}\right)\\
& = 2^{n(H(Y|V,Z) - I(V;Z)  -R  -4\d(\e))}.
\end{align*}}
$(a)$ follows since the set of $Y^n$ sequences are binned uniformly at random independent of other $Y^n$ sequences, and the fact that there are at most $2^{n(H(Y|Z,V)+ \d(\e))}$ $Y^n$ sequences which are jointly typical with a given typical $(v^n,z^n)$. $(b)$ follows from the fact that the codebook generation is independent of $(S^n,Z^n)$. Therefore, for any fixed $l$, $V^n(l)$ is independent of $Z^n$. Hence, if $R \ge I(V;S) - I(V;Z) + H(Y|Z,V) + 6\d(\e)$,
\begin{align*}
\sum_{l=1}^{2^{n(I(V;S) + \d(\e))}} \P(\Ec_l)  \le 2^{-n\d(\e)} \to 0,
\end{align*}
as $n \to \infty$.


We now turn to the proof of converse for Theorem \ref{thm:1}

\subsection*{Proof of converse for Theorem \ref{thm:1}}
Given a $(n,2^{nR})$ code for which the probability of error goes to zero with $n$ and satisfies the cost constraint, define $V_i = (Z^{n \backslash i }, S_{i+1}^n, Y^{i-1})$, we have
{\allowdisplaybreaks
\begin{align*}
   &nR  \\
   &\geq  H(M|Z^n)   \\
   &= H(M,Y^n|Z^n) - H(Y^n|M, Z^n)   \\
   &\stackrel{(a)}{=} H(M,Y^n|Z^n) - n\epsilon_n \\
   &= H(Y^n|Z^n) - n\epsilon_n\\
   &= \sum_{i=1}^n H(Y_i|Y^{i-1}, Z^n )- n\epsilon_n \\
   & =  \sum_{i=1}^n H(Y_i| Y^{i-1},S_{i+1}^n, Z^n) \\
   & \quad +\sum_{i=1}^n I(Y_i; S_{i+1}^n |Y^{i-1}, Z^n)- n\epsilon_n \\
   &\stackrel{(b)}{=} \sum_{i=1}^n H(Y_i| Y^{i-1},S_{i+1}^n, Z^n) \\
   &\quad  +\sum_{i=1}^n I(Y^{i-1}; S_i |S_{i+1}^n, Z^n)- n\epsilon_n \\
    & \stackrel{(c)}{=}\sum_{i=1}^n H(Y_i| Y^{i-1},S_{i+1}^n,Z^n) \\
    &\quad + \sum_{i=1}^nI(Y^{i-1}, S_{i+1}^n, Z^{n\backslash i}; S_i|Z_i)- n\epsilon_n\\
   & \stackrel{(d)}{=}\sum_{i=1}^n H(Y_i| V_i,Z^i) + \sum_{i=1}^nI(V_i; S_i|Z^i )- n\epsilon_n\\
   &= nH( Y_Q, |V_Q,Q,Z_Q) + nI(V_Q ; S_Q|Q,Z_Q)- n\epsilon_n
\end{align*}}
where (a) is due to Fano's inequality. (b) follows from Csisz\'{a}r Sum. (c) holds because $(S^n,Z^n)$ is an i.i.d source. Note that the Markov conditions, $V_i - (S_i,A_i) - Y_i$ and $V_i - S_i -Z_i$ hold. Finally, we introduce $Q$ as the time sharing random variable, i.e., $Q\sim\mathrm{Unif}[1,...,n]$, and  set $V=(V_Q,Q)$, $Y=Y_Q$ and $S=S_Q$, which completes the proof.

\textit{Remark:} Note that the proof of converse continues to hold even if side information $Z^n$ is available at the compressor. This observation shows that side information $Z^n$ at the compressor does not change the rate-cost tradeoff region in Theorem \ref{thm:1}. 
\subsection{Lossless, causal compression with action}
Our next result gives the rate-cost function for the case of lossless, causal compression with action.

\begin{theorem}[Rate-cost function for lossless, causal case] \label{thm:2} The rate for the compression with action when the state information is causally available at the action encoder is given by
\begin{equation}\label{eq.RC_causal}
R(B)=\min_{ p(v), a=f(s,v): E\Lambda(S,A,Y)\leq B} H(Y|V,Z)
\end{equation}
where the joint distribution is of the form $p(z,s,v,a,y) = p(z,s)p(v)1_{\{f(s,v)=a\}}p(y|a,s)$. The cardinality of $V$ is upper bounded by $|\Sc| + 2$. 
\end{theorem}
\textit{Achievability sketch}: Here $V$ simply serves as a time-sharing random variable. Fix a $p(v)$ and $f(s,v)$. We first generate a $V^n$ sequence and reveal it to the action encoder, the compressor and the decoder. The encoder generates $A_i = f(S_i,V_i)$. The compressor simply bins the set of $Y^n$ sequences to $2^{n(H(Y|V,Z) + \e)}$ bins and sends the index of the bin which contains $Y^n$. The decoder recovers $Y^n$ by finding the unique $Y^n$ sequence in bin $M$ such that $(V^n,Z^n, Y^n)$ are jointly typical.

\textit{Remark}: Just as in the non-causal case, the achievability is closely related to the channel coding strategy in \cite{GP1980}, our achievability in this section uses the ``Shannon Strategy'' in  \cite{Shannon1958}. In both cases, the optimal channel coding strategy yield the most compressible output when the message rate goes to zero.

\textit{Proof of Converse}: Given a $(n,2^{nR})$ code that satisfies the constraints, define $V_i = (S^{i-1}, Z^{n\backslash i})$. We have
{\allowdisplaybreaks
\begin{eqnarray}
nR &\geq & H(M|Z^n)   \nonumber \\
   &=& H(M,Y^n|Z^n) - H(Y^n|M, Z^n)   \nonumber \\
   &\stackrel{(a)}{=}& H(M,Y^n|Z^n) - n\epsilon_n \nonumber\\
   &=& H(Y^n|Z^n) - n\epsilon_n \nonumber\\
   &=& \sum_{i=1}^n H(Y_i|Y^{i-1}, Z_i, Z^{n\backslash i} )- n\epsilon_n \nonumber\\
   & \geq & \sum_{i=1}^n H(Y_i| Y^{i-1},A^{i-1},S^{i-1}, Z_i, Z^{n\backslash i})- n\epsilon_n \nonumber\\
   &\stackrel{(b)}{=}& \sum_{i=1}^n H(Y_i|A^{i-1},S^{i-1}, Z_i, Z^{n\backslash i})- n\epsilon_n \nonumber\\
   & \stackrel{(c)}{=}&\sum_{i=1}^n H(Y_i| V_i, Z_i)- n\epsilon_n \nonumber\\
   & \stackrel{(d)}{=}& n H(Y_Q| V_Q,Q, Z_Q)- n\epsilon_n \nonumber
\end{eqnarray}
}
where (a) is due to Fano's inequality; (b) follows from the Markov chain $Y_i - (S^{i-1},A^{i-1}, Z^n) - Y^{i-1}$ ; (c) follows since $A^{i-1}$ is a function of $S^{i-1}$. Note that $A_i$ is now a function of $S_i$ and $V_i$. Finally, we introduce $Q$ as the time sharing random variable, i.e., $Q\sim\mathrm{Unif}[1,...,n]$.  Thus, by setting $V=(V_Q,Q)$ and $Y=Y_Q$, we have completed the proof.

\textit{Remark:} Note that the proof of converse continues to hold even if side information $Z^n$ is available at the compressor. This observation shows that side information $Z^n$ at the compressor does not change the rate-cost tradeoff region in Theorem \ref{thm:2}. 

\subsection{Examples}
\subsubsection{No side information} In this subsection, we first consider an example with state sequence $S^n\sim$ i.i.d. Bern($1/2$) and $Z = \emptyset$. We have two actions available, $A=0$ and $A=1$. The cost constraint is on the frequency of action $A=1$, $EA\leq B$. The channel output $Y_i = S_i \oplus A_i \oplus S_{Ni}$ where $\oplus$ is the modulo 2 sum and $\{S_{Ni}\}$ are i.i.d. Bern$(p)$ noise, $p<1/2$.
The example is illustrated in Fig.~\ref{fig.ex1}.
\begin{figure}[ht]{
\psfrag{A}[][][.8]{$A^n$}
\psfrag{Y}[][][.8]{$Y^n$}
\psfrag{o}[][][1.2]{$+$}
\psfrag{A1}[][][.8]{Action}
\psfrag{E1}[][][.8]{Encoder}
\psfrag{C} [][][.8]{$EA\leq B$}
\psfrag{E2}[][][.8]{Compressor}
\psfrag{M}[][][.8]{$M\in$}
\psfrag{R}[][][.8]{$\{1,..,2^{nR}\}$}
\psfrag{Z}[][][.8]{$S_N^n\sim$ i.i.d Bern$(p)$}
\psfrag{S}[][][.8]{$\quad \quad\quad S^n\sim$ i.i.d Bern$(1/2)$}
\centerline{\includegraphics[width=3.3in]{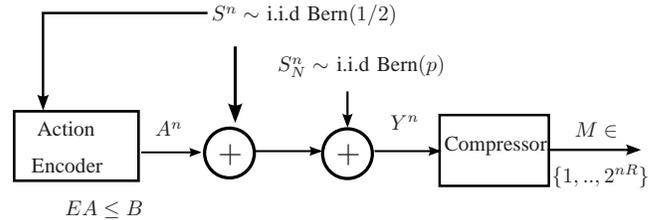}}}
\caption{Binary example with side information $Z = \emptyset$.}
\label{fig.ex1}
\end{figure}

We use the following lemma to simplify the optimization problem in Eq.~(\ref{eq.RC}) applied to the binary example.
\begin{lemma}\label{lemma.noncausal} For the binary example, it is without loss of optimality to have the following constraints when solving the optimization problem of Eq.~(\ref{eq.RC}):
\begin{itemize}
  \item $\mathcal{V}=\{0,1,2\}$, $\Pr(V=0)=\Pr(V=1)=\theta/2$, for some $\theta\in[0,1]$.
  \item The function $a=f(s,v)$ is of the form: $f(s,0)=s$, $f(s,1)=1-s$ and $f(s,2)=0$.
  \item $\Pr(S=0|V=1) = \Pr(S=1|V=0)=\Delta$ and $\Pr(S=0|V=2)=1/2$.
  \item $\Delta\theta\leq B$.
\end{itemize}
 \end{lemma}
Note that the constraints guarantee that $\Pr(S=0)=\Pr(S=1)=1/2$.

\begin{proof}
See Appendix.
\end{proof}
Using Lemma~\ref{lemma.noncausal}, we can simplify the objective function in Eq.~(\ref{eq.RC}) in the following way:{\allowdisplaybreaks
\begin{eqnarray}
&\quad& H(Y|V) + I(V;S)\nonumber\\
&=& H(Y|V) - H(S|V) + H(S)\nonumber\\
&=& H(S\oplus A \oplus S_N |V) - H(S|V) +1 \nonumber\\
&=&\frac{\theta}{2} \left( H( 0 \oplus S_N |V=0) - H(\Delta)\right)\nonumber\\
 &\quad&+\frac{\theta}{2} \left\{H( 1 \oplus S_N |V=1) - H(\Delta)\right\}\nonumber\\
 &\quad& +(1-\theta)\left\{H(S\oplus S_N|V=2)-1\right\}+1\nonumber\\
&=& \theta\left( H_2(p) - H(\Delta) \right)+1\nonumber
\end{eqnarray}}
where $H_2(\cdot)$ is the binary entropy function, i.e., $H_2(\delta) = -\delta\log\delta-(1-\delta)\log(1-\delta)$.

\begin{eqnarray}
R(B) &=&\min_{\theta\in[2B,1],\quad\theta\Delta\leq B} \theta\left( H_2(p) - H(\Delta) \right)+1\nonumber\\
&=&1+\min_{\Delta\in[B,1/2]} \frac{B}{\Delta}\left( H_2(p) - H_2(\Delta) \right)\nonumber\\
&=&1-B\max_{\Delta\in[B,1/2]} \frac{ H_2(\Delta) - H_2(p)}{\Delta}\nonumber\\
&=&\left\{
     \begin{array}{ll}
       1-B\frac{H(b^*)-H_2(p)}{b^*}, & \text{ if } 0\leq B < b^* \\
       1-H_2(B)+H_2(p), & \text{ if } b^* \leq B \leq 1/2
     \end{array}
   \right.
\end{eqnarray}
where $b^*$ is the solution of the following function:
\begin{equation}
\frac{H_2(b) - H_2(p)}{b} = \frac{dH_2}{db},\quad b\in[0,1/2]
\end{equation}
which is illustrated in Fig.~\ref{fig.entropy}.
\begin{figure}[h]{
\psfrag{A}[][][.8]{\ $A^n$}
\psfrag{p0}[][.8]{$(p,H(p))$}
\psfrag{bs}[][.8]{$b^*$}
\psfrag{p}[][.8]{$p$}
\psfrag{dH}[][.8]{$\quad\quad H(b^*)-H(p)$}
\psfrag{db}[][.8]{$b^*$}
\psfrag{f0}[][.7]{$\frac{H_2(b^*) - H_2(p)}{b^*} = \frac{dH_2}{db}\big|_{b=b^*}$}
\centerline{\includegraphics[width=3.5 in]{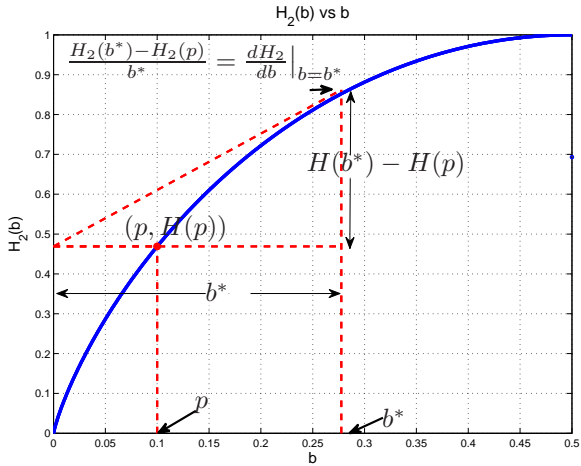}}}
\caption{The threshold $b^*$ solves $\frac{H_2(b) - H_2(p)}{b} = \frac{dH_2}{db}$,$\quad b\in[0,1/2]$}
\label{fig.entropy}
\end{figure}

Now let us shift our attention to the causal case of the binary example, i.e., $S_i$ is only causally available at the action encoder.
\begin{lemma}\label{lemma.causal} For the causal case of the binary example, it is without loss of optimality to have the following constraints when solving the optimization problem in Eq.~(\ref{eq.RC_causal}):
\begin{itemize}
  \item $\mathcal{V}=\{0,1\}$, $\Pr(V=0)=\theta$, for some $\theta\in[0,1]$.
  \item The function $a=f(s,v)$ is of the form: $f(s,0)=s$, $f(s,1)=0$.
  \item $\frac{\theta}{2}\leq B$.
\end{itemize}
 \end{lemma}
 \begin{proof}
 See Appendix.
 \end{proof}
{\allowdisplaybreaks
\begin{eqnarray}
&\quad&R(B)\nonumber\\
&=&\min H(Y|V)\nonumber\\
&=& \min_{\theta \in [0,1],\frac{\theta}{2}\leq B}  \theta H( Y|V=0)  + (1-\theta) H(Y|V=1)\nonumber\\
&=& \min_{\theta \in [0,1],\frac{\theta}{2}\leq B}  \theta H( Z|V=0)  + (1-\theta) H(S\oplus Z|V=1)\nonumber\\
&=& \min_{\theta \in [0,1],\frac{\theta}{2}\leq B}  \theta H_2(p) + (1-\theta)\nonumber\\
&=&\left\{
     \begin{array}{ll}
       2BH_2(p) + (1-2B), & \hbox{$0\leq B\leq 1/2$;} \nonumber\\
       H_2(p), & \hbox{$1/2\leq B$.}
     \end{array}
   \right.
\end{eqnarray}}
For the binary example with $p=0.1$, we plot the rate-cost function $R(B)$ for both cases in Figure \ref{fig.RB}. Note that when $S$ is only causally known at the action encoder, the optimum lossless compression scheme amounts to time sharing between compressing the noise $S_N$ losslessly and compressing $S$ losslessly. The optimum time sharing factor is determined by the cost $B$.
\begin{figure}[h]{
\centerline{\includegraphics[width=3.7in]{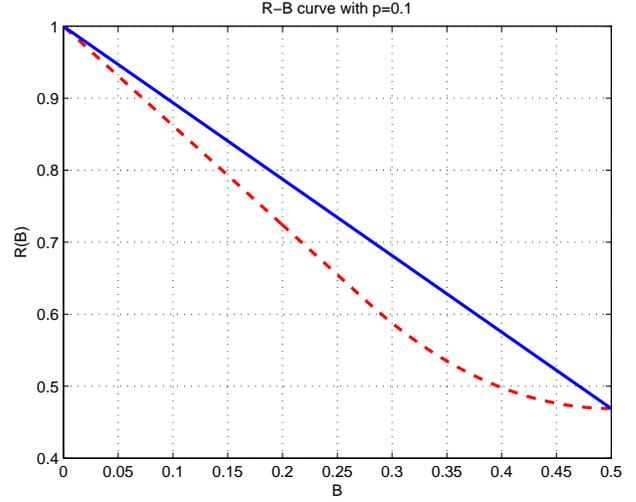}}}
\caption{Comparison between the non-causal and causal rate-cost functions. The parameter of the Bernoulli noise is set at 0.1.}
\label{fig.RB}
\end{figure}
{
\subsubsection{Erased side information} We now turn to the case when side information is available at the decoder only. We extend our setting in the previous example by letting $Z$ be an erased version of $S$. That is,
\begin{align*}
Z = \left\{\begin{array}{lcc}S & \mbox{w.p.} & 1 - p_e \\
e &\mbox{w.p.} & p_e\end{array}\right.
\end{align*}
In this case, the rate cost function is related to the case when no side information is available at the decoder in a simple manner. We first note the following
\begin{align}
\P(V|Z = e) &= \sum_{S} \P(S,V|Z = e) \nonumber\\
& = \sum \P(S|Z = e) \P(V|S,Z = e) \nonumber\\
&= \sum \P(S)\P(V|S) \nonumber\\
&= \P(V). \label{probeqn:1}
\end{align}
The third equality follows from the Markov Chain $V-S-Z$. Furthermore,
\begin{align}
\P(Y|V, Z = e) &= \sum_S p(S|Z = e, V) p(Y|V,S,Z = e) \nonumber\\
&= \sum P(S|V)\P(Y|V,S) \nonumber\\
& = \P(Y|V). \label{probeqn:2} 
\end{align}
The second equality follows from the Markov chain $Z - (S,V) - Y$ and $\P(S|Z=e, V) = \P(S,V,Z = e)/\P(V,Z = e) = \P(S,V)\P(Z = e)/(\P(Z = e)\P(v)) = \P(S|V)$. We now consider the rate-cost expression when $S$ is non-causally known at the action encoder.
\begin{align*}
R(B)_{\rm NC} &= \min I(V;S|Z) + H(Y|V,Z) \\
& \stackrel{(a)}{=} \min \, p_e I(V;S) + p_e H(Y|V) + (1- p_e) H(Y|V,S) \\
&= p_e \min (I(V;S) + H(Y|V) )+ (1- p_e) H_2(p).
\end{align*}
$(a)$ follows the following observations: (i) when $Z = \e$, $\P(V|Z = e) = \P(V)$ by \eqref{probeqn:1}, so $H(V|Z = e) = H(V)$ and from the Markov Chain $V-S-Z$ and $\P(S|Z = e) = \P(S)$, $H(V|S,Z = e) = H(V|S)$. Hence, $I(V;S|Z = e) = I(V;S)$; and (ii) from \eqref{probeqn:2}, $H(Y|V, Z = e) = H(Y|V)$. The last equality follows from when $Z = S$, $H(Y|Z = S, V) = H(Y|S,V)$. Since $A = f(S,V)$, $S_N \sim \Bern(p)$ independent of $(S,V,Z)$ and $Y = S \oplus A \oplus S_N$, $H(Y|S,V) = H_2(p)$.

As checks, note that when $p_e = 1$, which corresponds to the no side information case, the rate-cost function reduces to that in Corollary 1, and when $p_e = 0$, the rate-cost function reduces to $H_2(p)$, which corresponds to the minimum rate required when $S$ is also available at the decoder. We now turn to the case when $S$ is only causally known at the action encoder. Here, we have
\begin{align*}
R(B)_{\rm C} &= \min H(Y|V,Z) \\
& = p_e \min H(Y|V) + (1-p_e) H_2(p).
\end{align*}
The rate-cost tradeoff is shown in figure~\ref{fig.RBE}.

\begin{figure}[h]{
\centerline{\includegraphics[width=3.7in]{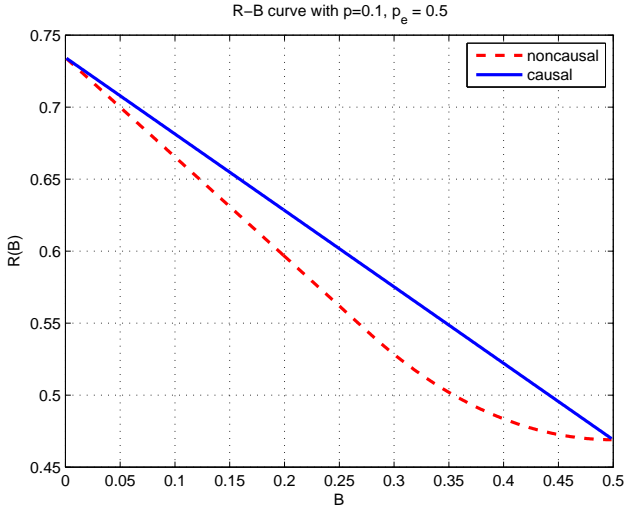}}}
\caption{Comparison between the non-causal and causal rate-cost functions with erased side information at the decoder. The parameter of the Bernoulli noise is set at 0.1 and the parameter of the erased side information is set at 0.5.}
\label{fig.RBE}
\end{figure}

}


%
\section{Lossy compression with actions} \label{sect:4}
{ In this section, we extend our setup to the lossy case. We first consider the case when side information is available at both the compressor and the decoder. We characterize the rate-distortion-cost tradeoff region for the case when $S^n$ is causally known to the action encoder. The case when $S^n$ is non-causally know at the action encoder is more involved. We give an achievable rate-distortion-cost region for that setting. We then move on to the case when side information $Z$ is available at the decoder only, and $S$ is non-causally known at the action encoder. We discuss two achievability schemes for this setting.
\subsection{Side information known at compressor and decoder}}
\begin{theorem} \label{thm:3} The rate-cost-distortion function for the case with causal state information and side information available to both the compressor and the decoder is given by
\begin{equation}\label{eq.causal.lossy}
R(B,D)=\min_{ a=f(s,v): E\Lambda(S,A,Y)\leq B, Ed(Y,\hat{Y})\leq D} I(Y;\Yh|V,Z)
\end{equation}
where the joint distribution is of the form $p(z,s,v,a,y,\hat{y}) = p(s,z)p(v)1_{\{a=f(s,v)\}}p(y|a,s)p(\yh|y,v,z)$. The cardinality of $V$ is upper bounded by $|\Sc| +2$.
\end{theorem}

\textit{Achievability sketch}: The achievability is straightforward, with $V^n$ acting as the time sharing random variable known to all parties. We first generate $V^n \sim \prod_{i=1}^n p(v_i)$. For each $z^n$ sequence, we generate $2^{n(I(\Yh;Y|V,Z)+ \e)}$ $\Yh^n$ sequences according to $\prod_{i=1}^n p(\yh_i|v_i,z_i)$. The action encoder simply generates $a^n$ according to $a_i = f(v_i,s_i)$ for $i \in [1:n]$. The compressor looks for a $\yh^n$ sequence such that $(\yh^n, y^n, v^n, z^n) \in \aep$. It then sends out this description to the decoder which reconstructs $Y^n$ as $\yh^n$. Since we have $2^{n(I(\Yh;Y|V,Z)+ \e)}$ $\Yh^n$ sequences, the probability of not finding a jointly typical $\Yh^n$ sequence goes to zero as $n \to \infty$.

\textit{Converse}: Given a $(n, 2^{nR})$ code satisfying the cost and distortion conditions, we have
{\allowdisplaybreaks
\begin{align*}
nR &\geq  H(M|Z^n)   \\
   &\geq I(M;Y^n|Z^n)    \\
   & =  \sum_{i=1}^n I(M; Y_i|Y^{i-1}, Z^{n\backslash i}, Z^i) \\
   & \stackrel{(a)}{=}  \sum_{i=1}^n I(M;Y_i |V_i, Z_i)\\
   & \stackrel{(b)}{\geq}  \sum_{i=1}^n I(\Yh_i ; Y_i |V_i,Z_i) \\
   & \stackrel{(c)}{=} n I(\hat{Y}_Q; Y_Q |V_Q, Q, Z_Q) 
\end{align*}
}
where in (a) we set $V_i = (Y^{i-1}, Z^{n \backslash i})$. (b) holds from the fact that $\Yh_i$ is a function of $M$ and $Z^n$. In (c) we introduce $Q$ as the time sharing random variable, i.e., $Q\sim\mathrm{Unif}[1,...,n]$. Thus, by setting $V=(V_Q,Q)$ and $Y=Y_Q$ and noting that Note that $V$ is independent of $(S, Z)$, we have shown that $R(B,D)\geq I(\hat{Y},Y|V,Z)$ for some $p(v)p(s|v)p(z|s)p(a|s,v)p(y|a,s)p(\yh|y,a,s,z,v)$. It suffices to restrict attention to the joint distribution stated in Theorem \ref{thm:3} because of the following observations.
\begin{itemize}
  \item $p(\hat{y}|a,y,v,s,z)$ can be restricted to $p(\yh|y,v,z)$ since the mutual information term $I(Y;\hat{Y}|V,Z)$ and the distortion constraint only depend on the marginal distribution $p(\hat{y},y,v,z)$.  
  \item $p(a|s,v)$ can be restricted to $a = f(s,v)$ since we can always find an independent random variable $U$ such that $p(a |s,v) = p(s,v)p(u)1_{a = f(s,v,u)}$. Now define $\Vt = (V,U)$ and $p(\yh|y,\vt, z) = p(\yh|y,\vt,z)$. Note that $p(s,z)p(\vt)1_{\{a=f(s,\vt)\}}p(y|a,s)p(\yh|y,\vt,z) = p(s,z)p(v)p(a|s,v)p(y|a,s)p(\yh|y,v,z)$. Since the joint distribution remains unchanged, the distortion and cost are preserved. As for the rate, we note that
\begin{align*}
I(\Yh;Y|\Vt,Z) & = H(\Yh|\Vt, Z) - H(\Yh|\Vt, Z, Y) \\
& = H(\Yh|V, U, Z) - H(\Yh|V, U,Z, Y) \\
& \le H(\Yh|V, Z) - H(\Yh|V,Z, Y) \\
& = I(\Yh;Y|V,Z). 
\end{align*}
\end{itemize}

Our next Theorem gives an upper bound on the rate-distortion-cost tradeoff for the case when the state information is known non-causally at the action encoder and $Z^n$ is present at both the compressor and the decoder.
\begin{theorem}An upper bound on the rate-distortion-cost function for the case with non-causal state information and side information at both the compressor and decoder is given by
\begin{equation}\label{eq.noncausal.lossy}
R(B,D)\leq \min_{E\Lambda(S,A,Y)\leq B,Ed(Y,\hat{Y})\leq D} I(V;S|Z)+I(\hat{Y};Y|V,Z)
\end{equation}
where the joint distribution is of the form $p(s,v,a,y,\hat{y},z) = p(s,z)p(v|s)1_{\{f(s,v)=a\}}p(y|a,s)p(\hat{y}|y,v,z)$.
\end{theorem}
\textit{Sketch of achievability}:

We generate $2^{n(I(V;S) +\e)}$ $V^n(l_0)$, $l_0\in [1:2^{n(I(V;S) +\e)}]$, sequences according to $\prod_{i=1}^n p(v_i)$, and for each $v^n(l_0)$ and $z^n$, generate $2^{n(I(\Yh;Y|V,Z)+\e)}$ $\Yh^n(l_0, l_1)$, $l_1 \in [1:2^{n(I(\Yh;Y|V,Z)+\e)}]$, sequences according to $\prod_{i=1}^n p(\yh_i|v_i,z_i)$. The set of $v^n$ sequences are then randomly binned to $2^{n(I(V;S|Z) +2\e)}$ bins, $\Bc(m)$, $m \in [1:2^{n(I(V;S|Z) +2\e)}]$. Given a sequence $s^n$, the action encoder finds the $v^n$ sequence which is jointly typical with $s^n$ and takes actions according to $a_i = f(s_i,v_i)$ for $i \in [1:n]$. At the compressor, we first find a $\vt^n(l_0)$ that is jointly typical with $(y^n, z^n)$ and then, a $\yh^n(l_0,l_1)$ such that $(\vh^n(l_0), \yh^n(l_0,l_1), y^n,z^n) \in \aep$. Note that there exists at least one $\vt^n(l_0)$ that is jointly typical with $(y^n,z^n)$ with high probability since the true $v^n$ sequence is jointly typical with $(y^n,z^n)$ with high probability. If there is more than one such sequence, the compressor chooses one uniformly at random from the set of $\vt^n$ sequences jointly typical with $(y^n, z^n)$. The compressor then sends the indices $m$ and $l_1$ such that the selected $\vt^n(l_0) \in \Bc(m)$. The decoder recovers $\vt^n(l_0)$ by looking for the unique $\lh_0 \in \Bc(m)$ such that $(\vh^n(\lh_0), z^n) \in \aep$. It reconstructs $Y^n$ as $\yh(\lh_0, l_1)$. From the rates given, it is easy to see that all encoding and decoding steps succeed with high probability as $n \to \infty$.

{
\subsection{Side information available at the decoder only}
When the side information is available at the decoder only and $S$ is know non-causally at the action encoder, we discuss two possible achievability schemes. The first scheme is a generalization of the achievability scheme of Theorem 1 to the lossy case. 
\begin{theorem} \label{thm:5}
An upper bound on the rate-distortion-cost function for the case with non-causal state information and side information at the decoder is given by
\begin{align*}
R(B,D) = I(V;S|Z) + I(U;Y|V,Z)
\end{align*}
for some $p(z,s)p(v|s)1_{a = f(s,v)}p(y|a,s)p(u|y)$ satisfying
\begin{align*}
\E \Lambda(A,S,Y) &\le B, \\
\E d(Y, \Yh(Z,U)) & \le D.
\end{align*}
\end{theorem}
Theorem \ref{thm:5} is generalization of Theorem 1 since if we let $U = Y$, we recover Theorem 1. 

\begin{proof}
As the achievability scheme is an extension of the achievability scheme for Theorem 1, we will only mention the additional steps in the proof of achievability. 
\subsection*{Codebook Generation}  
The additional step in the codebook generation procedure is in the generation of a codebook of $U^n$ covering sequences to cover  $Y^n$ and a binning or compression codebook for the $ U^n$ sequences. We first generate $2^{n(I(U;Y) + \e)}$ $U^n(l_1)$ sequences according to $\prod_{i=1}^n p(u_i)$. We then bin the set of $U^n$ sequences into $2^{n (I(V;S|Z) + I(U;Y|V,Z) + 5\e)}$ bins, $\Bc(M)$, $M \in [1:2^{n (I(V;S|Z) + I(U;Y|V,Z) + 5\e)}]$.
\subsection*{Encoding}
The encoding procedure for the action encoder remains the same as that in Theorem 1. For the compressor, it first looks for a $u^n(l_1)$ such that $(u^n(l_1), y^n) \in \aep$. It then sends out the index $m$, such that $u^n(l_1) \in \Bc(m)$.
\subsection*{Decoding and analysis of probability of error}
The decoder looks for the unique $u^n(\lh_1) \in \Bc(m)$ such that $(v^n(l_0), y^n,z^n, u^n(l_1)) \in \aep$ for some $l_0 \in [1:2^{n(I(V;S) + \e}]$. For the analysis of probability of error, let $L_0$ and $L_1$ be the indices picked by the action encoder and the compressor respectively. Following the rates given in the codebook generation and encoding procedure, the covering lemma \cite[Chapter 3]{El-Gamal--Kim2010} and the strong Markov lemma \cite[Chapter 12]{El-Gamal--Kim2010}, it is easy to see that $\P(V^n(L_0), Y^n, ,Z^n, U^n(L_1) \in \aep) \to 1$ as $n \to \infty$. The other ``error'' event of interest is now the following.
\begin{align*}
\Ec_U &:= \{(V^n(l_0), Z^n, U^n(\lh_1)) \in \aep \\
& \qquad \mbox{ for some } \lh_1 \neq L_1, \U^n(\lh_1) \in \Bc(M) ,\\
& \qquad l_0 \in [1:2^{n(I(V;S) + \e)}]\}.
\end{align*}
Due to the symmetry of the binning process, we can assume without loss of generality, $M = 1$. Define $\Ec_{l_0}(V^n,Z^n)$ to be the event
\begin{align*}
\Ec_{l_0}(V^n,Z^n) &:= \{(V^n(l_0), Z^n, U^n(\lh_1)) \in \aep \\
& \qquad \mbox{ for some } \lh_1 \neq L_1, \U^n(\lh_1) \in \Bc(1)\}.
\end{align*}
Then, $\P(E_U)$ is upper bounded by 
\begin{align}
\P(E_U) &\le \sum_{l_0} \P(\Ec_{l_0}(V^n,Z^n)). \label{eqnerror1}
\end{align}
We now give a bound for $\P(\Ec_{l_0}(V^n,Z^n))$. We first have
\begin{align*}
&\P(\Ec_{l_0}(V^n,Z^n)) \\
& = \sum_{(v^n,2^{nI(V;S)}z^n) \in \aep}\left(\begin{array}{ll}p_{V,Z}(v^n,z^n)\mathbf{.}\\ \P(\Ec_{l_0}(V^n,Z^n)|V^n = v^n, Z^n = z^n)
\end{array}\right) \\
& = \sum_{(v^n,z^n) \in \aep}\left(\begin{array}{ll}p_{V,Z}(v^n,z^n)\mathbf{.}\\ \P(\Ec_{l_0}(v^n,z^n)|v^n, z^n)
\end{array}\right).
\end{align*}

Next, let $\Ec_{l_0}(j, v^n,z^n)$ be the error event $\{(v^n(l_0), z^n, U^n(j)) \in \aep, L_1 \neq j, U^n(j) \in \Bc(1)\}.$ Then, $\Ec_{l_0}(v^n,z^n) \subseteq \cup_{j=1}^{2^{n(I(U;Y) + \e)}}\Ec_{l_0}(j,v^n,z^n)$. We therefore have
\begin{align*}
& \P(\Ec_{l_0}(v^n,z^n|v^n,z^n)) \\
& \le \sum_{j=1}^{2^{n(I(U;Y) + \e)}} \P(\Ec_{l_0}(j,v^n,z^n)|v^n,z^n) \\
& \le \sum_{j=1}^{2^{n(I(U;Y) + \e)}} \P((v^n, z^n, U^n(j)) \in \aep, U^n(j) \in \Bc(1)|v^n,z ^n) \\
& = \sum_{j=1}^{2^{n(I(U;Y) + \e)}} \left(\begin{array}{ll}\P((v^n, z^n, U^n(j)) \in \aep|v^n,z^n)\mathbf{.}\\ \P(U^n(j) \in \Bc(1)|v^n,z^n)\end{array}\right) \\
& \le 2^{n(I(U;Y) + \e)}.2^{-n(I(U;Z,V) -\e)}. 2^{-n(I(V;S|Z) + I(U;Y|V,Z) + 5\e)}\\
& = 2^{-n(I(V;S|Z)+3\e)}. 
\end{align*}

We therefore have
\begin{align*}
&\P(\Ec_{l_0}(V^n,Z^n)) \\
&\le \sum_{v^n,z^n \in \aep} \P(V^n(l_0) = v^n, Z^n = z^n)\mathbf{.}2^{-n(I(V;S|Z)+3\e)} \\
& \le 2^{-n(I(V;Z)-\e)}\mathbf{.}2^{-n(I(V;S|Z)+3\e)} .
\end{align*}
Hence, from \eqref{eqnerror1},
\begin{align*}
\P(E_U) &\le 2^{n(I(V;S) + \e)}2^{-n(I(V;Z)-\e)}\mathbf{.}2^{-n(I(V;S|Z)+3\e)} \\
& = 2^{-n\e}.
\end{align*}
Therefore, $\P(E_U) \to 0$ as $n \to \infty$.

Since the probability of ``error'' goes to zero as $n \to \infty$, the expected distortion of the reconstruction $\Yh_i = \yh(z_i, u_i)$, $i \in [1:n]$, is less than or equal to $D$ as $n \to \infty$.
\end{proof}
The achievability scheme in Theorem \ref{thm:5} restricts the description of $Y^n$ that is sent, $U^n$, to be independent of $V^n$ given $Y^n$. This is a result of not requiring the compressor to decode the true $V^n$ codeword that was selected by the action encoder. In our next scheme, we remove the Markov condition, $U-Y-V$, by making the compressor decode $V^n$. This operation results in a different restriction on the allowable joint probability distribution, $I(V;Y) \ge I(V;S)$. 

\begin{theorem} \label{thm:6}
An upper bound on the rate-distortion-cost function for the case with non-causal state information and side information at the decoder is given by
\begin{align*}
R(B,D) = I(V;S|Z) + I(U;Y|V,Z)
\end{align*}
for some $p(z,s)p(v|s)1_{a = f(s,v)}p(y|a,s)p(u|y,v)$ satisfying
\begin{align*}
\E \Lambda(A,S,Y) &\le B, \\
\E d(Y, \Yh(Z,U)) & \le D, \\
I(V;S) &\le I(V;Y).
\end{align*}
\end{theorem}

\noindent \textit{Sketch of achievability}

We generate $2^{n(I(V;S) +\e)}$ $V^n(l_0)$, $l_0\in [1:2^{n(I(V;S) +\e)}]$, sequences according to $\prod_{i=1}^n p(v_i)$, and for each $v^n(l_0)$, generate $2^{n(I(U;Y|V)+\e)}$ $U^n(l_0, l_1)$, $l_1 \in [1:2^{n(I(U;Y|V,Z)+\e)}]$, sequences according to $\prod_{i=1}^n p(u_i|v_i)$. The set of $V^n$ sequences are partitioned to $2^{n(I(V;S|Z) +2\e)}$ bins, $\Bc(m_0)$, $m_0 \in [1:2^{n(I(V;S|Z) +2\e)}]$, while the set of $U^n$ sequences are partitioned to $2^{n(I(U;Y|V,Z) + 2\e)}$ bins, $\Bc(m_0)$, $m_0 \in [1:2^{n(I(U;Y|V,Z) + 2\e)}]$. Given a sequence $s^n$, the action encoder finds the $v^n$ sequence which is jointly typical with $s^n$ and takes actions according to $a_i = f(s_i,v_i)$ for $i \in [1:n]$. At the compressor, we first find $v^n(l_0)$ by joint typicality decoding. It can be shown that this decoding procedure succeeds with high probability provided $I(V;Y) \ge I(V;S)$ \cite{YHK}. Next, the compressor looks for a $u^n(l_0, l_1)$ such that $(v^n(l_0), u^n(l_0, l_1), y^n, z^n) \in \aep$. The compressor then sends the indices $m_0$ and $m_1$ such that $v^n(l_0) \in \Bc(m_0)$ and $u^n(l_0, l_1) \in \Bc(m_1)$. The decoding operation now follows standard Wyner-Ziv decoding. The decoder first recovers $v^n(l_0)$ by looking for the unique $\lh_0 \in \Bc(m_0)$ such that $(\vh^n(\lh_0), z^n) \in \aep$. Next, it recovers $u^n(l_0, l_1)$ by looking for the unique $\uh^n(\lh_0,\lh_1)$ such that $(\vh^n(\lh_0), \uh^n(\lh_0,\lh_1), z^n) \in \aep$. It reconstructs $Y^n$ as $\yh_i(\vh_{i}(\lh_0), \uh_i(\lh_0, \lh_1), z_i)$ for $i \in [1:n]$. From the rates given, the encoding and decoding steps succeed with high probability as $n \to \infty$.

\section{Conclusion and future directions} \label{sect:5}
In this paper, we consider a variation of lossy and lossless compression where, instead of compressing the original source, we take actions to modify the source before compression, subject to a cost constraint. In the lossless case, we characterize the rate-cost tradeoff for several different cases, including the cases where side information is available at the decoder only, and where the original source $S^n$ is either known causally or non-causally at the action encoder. We then extended the analysis to the lossy case, where we characterize the rate-distortion-cost tradeoff for the case where $S^n$ is know only causally and side information is available at both the compressor and the decoder.

Our setting can be extended in several different directions. One possible extension is to consider the case of message embedding, where we desire to send a message together with conveying information about $Y^n$. Another extension that may be of interest is to consider the case where we have distributed state information $S_1$ and $S_2$ which are correlated at two different action encoders. We are still interested in the output $Y$, but the additional dimension in this extension is in how the two distributed action encoders can coordinate to generate an output $Y^n$ that is as compressible as possible. 
}
\section{Acknowledgement}
This work is supported in part by the National Science Foundation (NSF) through Grant 0939370-CCF and in part by Air Force Office of Scientific Research (AFOSR) through Grant FA9550-10-1-0124. We thank Professor Abbas El Gamal for helpful discussion.

\begin{appendix}
\subsection{Proof of Lemma~\ref{lemma.noncausal}}
Fixing a $v$, the function $a=f(s,v)$ has only four possible forms: $a=s$, $a=1-s$, $a=0$ and $a=1$. Thus, we can divide $\cV$ into four groups:
\begin{eqnarray}
\cV_0 &=&\{v: f(s,v) = s\}\nonumber\\
\cV_1 &=&\{v: f(s,v) = 1-s\}\nonumber\\
\cV_2 &=&\{v: f(s,v) = 0\}\nonumber\\
\cV_3 &=&\{v: f(s,v) = 1\}
\end{eqnarray}
First, it is without loss of optimality to set $\cV_3 =\emptyset$. That is because for each $v\in\cV_3$, we can change the function to $f(s,v)=0$. The rate $I(V;S)+H(Y|V)$ does not change and the cost $EA$ only decreases.

Rewrite the objective function in the following way
\begin{eqnarray}
&\quad&I(V;S)+H(Y|V) \label{eq.obj_binary}\\
&=& H(Y|V) - H(S|V)+ H(S)\nonumber\\
&=& H(S\oplus A\oplus Z|V) - H(S|V)+ H(S)\nonumber\\
&=& \sum_{v\in\cV_0} \big(H_2(p) - H(S|V=v) \big)p(v)   \nonumber\\
&\quad&\quad +  \sum_{v\in\cV_1} \big( H_2(p)-H(S|V=v)\big) p(v) \nonumber\\
&\quad&\quad +  \sum_{v\in\cV_2} \big(H( S \oplus S_{N} |V=v)-H(S|V=v)\big) p(v)\nonumber
\end{eqnarray}
where the last step is obtained by plugging in the actual form of $a=f(s,v)$ for each group of $v$.

Second, it is sufficient to have $|\cV_0|=1$ and $|\cV_1|=1$. To prove this, let $v_1,v_2\in \cV_0$. Note that $H(S|V=v)$ is a concave function in $p(s|V=v)$. Thus if we replace $v_1,v_2$ by a $v_3$ with $p(v_3)=p(v_1)+p(v_2)$ and
\begin{eqnarray*}
p(s|V=v_3)
&=& \frac{p(v_1)}{p(v_1)+p(v_2)}p(s|V=v_1)\\
&&\quad\quad+\frac{p(v_2)}{p(v_1)+p(v_2)}p(s|V=v_2),
\end{eqnarray*}
we preserve the distribution of $S$, the cost $EA$ but we reduce the first term, i.e., $\sum_{v\in\cV_0} \big(H_2(p) - H(S|V=v) \big)p(v)$,  in Eq.~(\ref{eq.obj_binary}).
Therefore, we can set $\cV_0 = \{0\}$ and $\cV_1 = \{1\}$.

Third note that for each $v\in\cV_2$,
\begin{eqnarray}
&\quad&H(Y|V=v) - H(S|V=v)\nonumber\\
&=& H(S\oplus A \oplus Z |V=v)- H(S|V=v)\nonumber\\
&=&  H(S\oplus S_N |V=v)- H(S|V=v)\nonumber\\
&\geq& 0
\end{eqnarray}

Last, if $\Pr(S=0|V=0)\neq \Pr(S=1|V=1)$, consider a new auxiliary random variable $V'$ with the following distribution:
\begin{itemize}
  \item $\mathcal{V'}=\{0,1,2\}$, $\Pr(V'=0)=\Pr(V'=1)=(\Pr(V=0)+\Pr(V=1))/2$
  \item The function $a=f(s,v')$ is of the form: $f(s,0)=s$, $f(s,1)=1-s$ and $f(s,2)=0$.
  \item $\Pr(S=0|V'=2)=1/2$ and
  \end{itemize}
        \begin{eqnarray*}
        &\quad&\Pr(S=1|V'=0)=\Pr(S=0|V'=1) \\
        &=& \frac{\Pr(S=1|V=0)\Pr(V=0)+\Pr(S=0|V=1)\Pr(V = 1)}{\Pr(V=0)+\Pr(V=1)}.
        \end{eqnarray*}

Comparing $(S,V')$ with $(S,V)$, we can check that the cost $EA$ and the distribution of $S$ are preserved. Meanwhile, the objective function is reduced, which completes the proof.

\subsection{Proof of Lemma~\ref{lemma.causal}}

Similar to the proof of Lemma~\ref{lemma.noncausal}, we divide $\cV$ in to $\cV_0,\cV_1,\cV_2,\cV_3$. Using the same argument, we show that $\cV_3=\emptyset$. Rewrite the objective function $H(Y|V)$ in the following way:
\begin{eqnarray}
&\quad&H(Y|V) \label{eq.obj2_binary}\\
&=& H(S\oplus A\oplus S_N|V) \nonumber\\
&=& \sum_{v\in\cV_0} H_2(p) p(v)   \nonumber\\
&\quad&\quad +  \sum_{v\in\cV_1} H_2(p) p(v) \nonumber\\
&\quad&\quad +  \sum_{v\in\cV_2} \big(H( S \oplus S_N |V=v)p(v)\nonumber\\
&=& H_2(p)\sum_{v\in\cV_0 \bigcup \cV_1}p(v) + \sum_{v\in\cV_2}p(v),\nonumber
\end{eqnarray}
which implies that it is sufficient to consider the case $|\cV_0|=1$, $\cV_1=\emptyset$ and $|\cV_2|=1$. And this completes the proof.
\end{appendix}
\end{document}